\documentclass{mfa}

\usepackage{mathtools}
\usepackage{amsmath}
\usepackage{amsfonts}
\usepackage{latexsym}
\usepackage{color}
\usepackage{graphicx}
\usepackage{subcaption}
\usepackage{units}
\usepackage{listings}
\usepackage{courier}
\usepackage{verbatim}
\usepackage{bigints}
\graphicspath{ {/home/gabe/schoolstuff/BDA} }
\newcommand\givenbase[1][]{\:#1\lvert\:}
\let\given\givenbase

\DeclarePairedDelimiterX\Basics[1](){\let\given\sgiven #1}

\newcommand{\diff}{\mathrm{d}}

\usepackage{scalerel,stackengine}
\stackMath
\newcommand\reallywidehat[1]{\savestack{\tmpbox}{\stretchto{  \scaleto{    \scalerel*[\widthof{\ensuremath{#1}}]{\kern-.6pt\bigwedge\kern-.6pt}    {\rule[-\textheight/2]{1ex}{\textheight}}  }{\textheight}}{0.5ex}}\stackon[1pt]{#1}{\tmpbox}}
\begin{document}

\title[Hierarchical Bayesian Bradley-Terry for Major League Baseball]
{Hierarchical Bayesian Bradley-Terry\\
  for Applications in Major League Baseball}
\author[Phelan]{Gabriel C. Phelan}
\address{School of Mathematical Sciences, Rochester Institute of Technology, 85 Lomb Memorial Drive, Rochester, New York 14623}
\email{gxp3900@rit.edu}

\author[Whelan]{John T. Whelan}
\address{School of Mathematical Sciences, Rochester Institute of Technology, 85 Lomb Memorial Drive, Rochester, New York 14623}
\email{jtwsma@rit.edu}

{62F07}

\setlength\parindent{0pt}
\setlength{\overfullrule}{0pt}

\newtheorem*{theorem}{Theorem}
\newtheorem*{lemma}{Lemma}
\newtheorem*{example}{Example}
\newtheorem*{assumption}{Assumption}
\newtheorem*{proposition}{Proposition}

\begin{abstract}
 A common problem faced in statistical inference is drawing conclusions from paired comparisons, in which two objects compete and one is declared the victor. A probabilistic
approach to such a problem is the Bradley-Terry model \cite{BT, Zermelo}, first studied by Zermelo in 1929 and rediscovered by Bradley and Terry in 1952. One obvious area of application
  for such a model is sporting events, and in particular Major League Baseball. With this in mind, we describe a hierarchical Bayesian 
  version of Bradley-Terry suitable for use in ranking and prediction problems, and compare results from these application domains to standard maximum likelihood approaches.  
  Our Bayesian methods outperform the MLE-based analogues, while being simple to construct, implement, and interpret.
\end{abstract}

\maketitle

\section{Background}
\subsection{The Bradley-Terry Model}
Amongst a set of $N$ objects, which we will call ``teams'',  the Bradley-Terry model associates a ``strength'' $\pi_i \in \mathbb{R}^+$ to each team and assumes that
\begin{align}
\begin{split}
  \mathbb{P}\{\text{team }i  \text{ defeats team } j\} =\frac{\pi_i}{\pi_i + \pi_j}, 
\end{split}
\end{align}
where $i,j \in \{1,2,\ldots,N\}$. If we define $V_{ij}$ to be the number of times in a season that team $i$ defeats team $j$, and $n_{ij}$ the number of games between them, an entire season can be described in
terms of the probability mass function
\begin{equation}
  p(\mathbf{V} \given \boldsymbol{\pi}) = \prod_{i=1}^{N-1} \prod_{j=i+1}^N \binom{n_{ij}}{V_{ij}} \left ( \frac{\pi_i}{\pi_i + \pi_j} \right )^{V_{ij}} \left ( \frac{\pi_j}{\pi_i + \pi_j} \right )^{V_{ji}}, 
\end{equation}
where $\mathbf{V}_{N \times N}$ is a matrix of records and $\boldsymbol{\pi}_{N \times 1}$ a vector of team strengths.
\subsection{The Bayesian Approach}
In practice, it is often the case that $\mathbf{V}$ is known and the goal is to
perform inference on the strengths $\boldsymbol{\pi}$. In a frequentist interpretation, this would proceed by maximum likelihood estimation, for which there is no closed-form solution
but a number of numerical algorithms have been suggested \cite{Csiszar, Ford, Hunter}. However, as noted by Ford \cite{Ford}, pathologies may exist under this approach. Maximum likelihood
ratios of some teams' strengths may
be zero, infinite, or undetermined, leading to 0, 1, or undetermined
probabilities.
The conditions under which these pathologies arise has been studied by various researchers \cite{AA, WB, SD}. Settings like Major League Baseball (unlike, for instance, college football) are practically guaranteed immunity from these issues \cite{WB}, but taking
a Bayesian approach fully guards against them.\\ \\
In a Bayesian interpretation of the model, we place a prior distribution $p(\boldsymbol{\pi})$ over the strengths, avoiding the aforementioned difficulties. It
also affords us the use of full Bayesian inference, in which we compute a posterior distribution $p(\boldsymbol{\pi} \given \mathbf{V})\propto p(\mathbf{V} \given \boldsymbol{\pi})p(\boldsymbol{\pi})$
over the team strengths in light of the records. $p(\boldsymbol{\pi} \given \mathbf{V})$ captures and quantifies all of our uncertainty conditional on our knowledge. Various takes on Bayesian Bradley-Terry have been studied \cite{Caron, Chen, DS, Leonard, Whelan}; a common concern is
the choice of prior distribution. In this regard we draw on the work of Whelan, who advocates for two classes of distributions in particular \cite{Whelan}.

\subsection{Desiderata and Choice of Prior Distribution}
In specifying a prior for Bayesian Bradley-Terry, one approach is to require that it adheres to a list of \textit{desiderata}. These formalize our intuition about how the model should 
behave under a suitable prior distribution. We adopt the desiderata of Whelan \cite{Whelan}, which, with applications to ranking systems in mind,
attempts to construct priors that make no unfair distinctions between individual teams. Roughly speaking, this means that
we should choose a prior, possibly over a transformed parameter $\mathbf{T}(\boldsymbol{\pi})$, that:
\vskip 3mm
\label{desi}
\begin{center}
  \begin{enumerate}
  \item Ensures invariance under the interchange of teams.
  \item Ensures invariance under the interchange of winning and losing.
  \item Ensures invariance under the elimination of teams.
    \item Is a proper (normalizable) prior. 
  \end{enumerate}
\end{center}
\vskip 3mm
Most families of prior distribution fail at least one of these requirements, but there are two families that are known to satisfy all four \cite{Whelan}.
The first is a separable Gaussian
distribution in the log-strengths with 0 mean and common variance:
\label{prior:gauss}
\begin{align}
  \begin{split}
  \lambda_i \sim \mathcal{N}(0, \sigma^2), \hspace{5mm} \lambda_i = T_1(\pi_i) =\log \pi_i. 
  \end{split}
\end{align}
The second is a Beta distribution in what can be interpreted as the probability of a particular team defeating an ``imaginary opponent'' of unit strength, with common scale and shape parameters:
 \label{prior:beta}
\begin{equation}
\zeta_i \sim \beta(\eta, \eta), \hspace{5mm} \zeta_i = T_2(\pi_i)=\frac{\pi_i}{1+\pi_i}.
\end{equation}
It is clear based on these definitions that $\boldsymbol{\lambda} \in \mathbb{R}^N$ and $\boldsymbol{\zeta} \in (0,1)^N$. We consider only these two families of prior distribution, denoted
as $I_{\mathcal{N}}$ and $I_{\beta}$, guaranteeing the desiderata are satisfied.
As we proceed, we will estimate posterior densities using Markov Chain Monte Carlo (MCMC), which is known to prefer unconstrained parameter spaces \cite{Stan}. It is therefore helpful to transform
the prior on $\boldsymbol{\zeta}$ into $\boldsymbol{\lambda}$-space. This parameterization will also ease any comparisons we may wish to make between the two priors.

\begin{lemma} $\lambda_i \given I_{\beta}$ has the generalized logistic distribution of the third kind, which we denote as $\lambda_i \given I_{\beta} \sim \mathrm{GL}_3(\eta)$.
  \label{changevars}
\end{lemma}
\begin{proof}
For each $\zeta_i$ we have
\begin{equation}
  p(\zeta_i \given I_\beta ) \propto [\zeta_i(1-\zeta_i)]^{\eta-1}
\end{equation}
and
\begin{equation}
  \zeta_i = \frac{e^{\lambda_i}}{1+e^{\lambda_i}}=\mathrm{logistic}(\lambda_i).
\end{equation}
Thus,
\begin{equation}
  \left | \frac{\diff \zeta_i}{\diff \lambda_i} \right |  = \frac{e^{\lambda_i}}{(1+e^{\lambda_i})^2},
\end{equation}
and by change of variables,
\begin{equation}
  \begin{split}
  p(\lambda_i \given I_{\beta}) & \propto \left [
       \frac{e^{\lambda_i}}{1+e^{\lambda_i}} \left ( 1- \frac{e^{\lambda_i}}{1+e^{\lambda_i}} \right )
       \right ]^{\eta-1}
  \frac{e^{\lambda_i}}{(1+e^{\lambda_i})^2} \\
  & \propto \left [ \frac{e^{\lambda_i}}{(1+e^{\lambda_i})^2} \right ]^{\eta-1} \frac{e^{\lambda_i}}{(1+e^{\lambda_i})^2} \\
  & \propto \left [ \frac{e^{\lambda_i}}{(1+e^{\lambda_i})^2} \right ]^\eta,
  \end{split}
\end{equation}
which is the form of a $\mathrm{GL_3}$ distribution. 
\end{proof}
This family of distributions is well-known, the most general form of which is given by
\begin{align}
  \begin{split}
  p(\lambda_i \given \varphi, \eta, \gamma) = \frac{1}{\mathrm{B}(\gamma, \eta)} \left ( \frac{\varphi e^{-\varphi \eta \lambda_i}}{(1+e^{-\varphi \lambda_i})^{\eta+\gamma}} \right ),
  \end{split}
\end{align}
where $\mathrm{B}(\gamma,\eta) = \frac{\Gamma(\gamma) \Gamma(\eta)}{\Gamma(\gamma + \eta)}$ is the Beta function. We say that $\lambda_i \sim \mathrm{GL}(\varphi, \eta, \gamma)$.
A complete overview is given in \cite{NE}, where the following useful properties are shown:
\begin{equation}
  \begin{split}
    \mathbb{E}[\lambda_i \given \varphi, \eta, \gamma] & = \frac{1}{\varphi}[\psi(\gamma)-\psi(\eta)] \\
    \mathbb{V}[\lambda_i \given \varphi, \eta, \gamma] & = \frac{1}{\varphi^2}[\psi^{\prime}(\gamma)+\psi^{\prime}(\eta)],
  \end{split}
\end{equation}
for $\psi(\cdot)$ and $\psi^{\prime}(\cdot)$ the digamma and trigamma functions respectively.   
It is evident that the $\mathrm{GL}_3(\eta)$ distribution is equivalent to the $\mathrm{GL}(1,\eta,\eta)$ distribution, and so we immediately have
\begin{equation}
    \mathbb{E}[\lambda_i \given I_{\beta}]  = 0 \hspace{5mm}\text{and}\hspace{5mm}
    \mathbb{V}[\lambda_i \given I_{\beta}]  = 2\psi^{\prime}(\eta).
\end{equation}
For reference, $\eta=1$, which would produce a uniform prior in $\zeta_i$, corresponds to a Gaussian-like prior with variance $2\psi^{\prime}(1) \approx 3.3$ in $\lambda_i$.
With $p(\lambda_i \given I_{\beta})$ established, we restrict our discussion to $\boldsymbol{\lambda}$-space from here onward. 
\begin{figure}
    \centering
    \begin{subfigure}[b]{0.5\textwidth}
        \includegraphics[width=\textwidth]{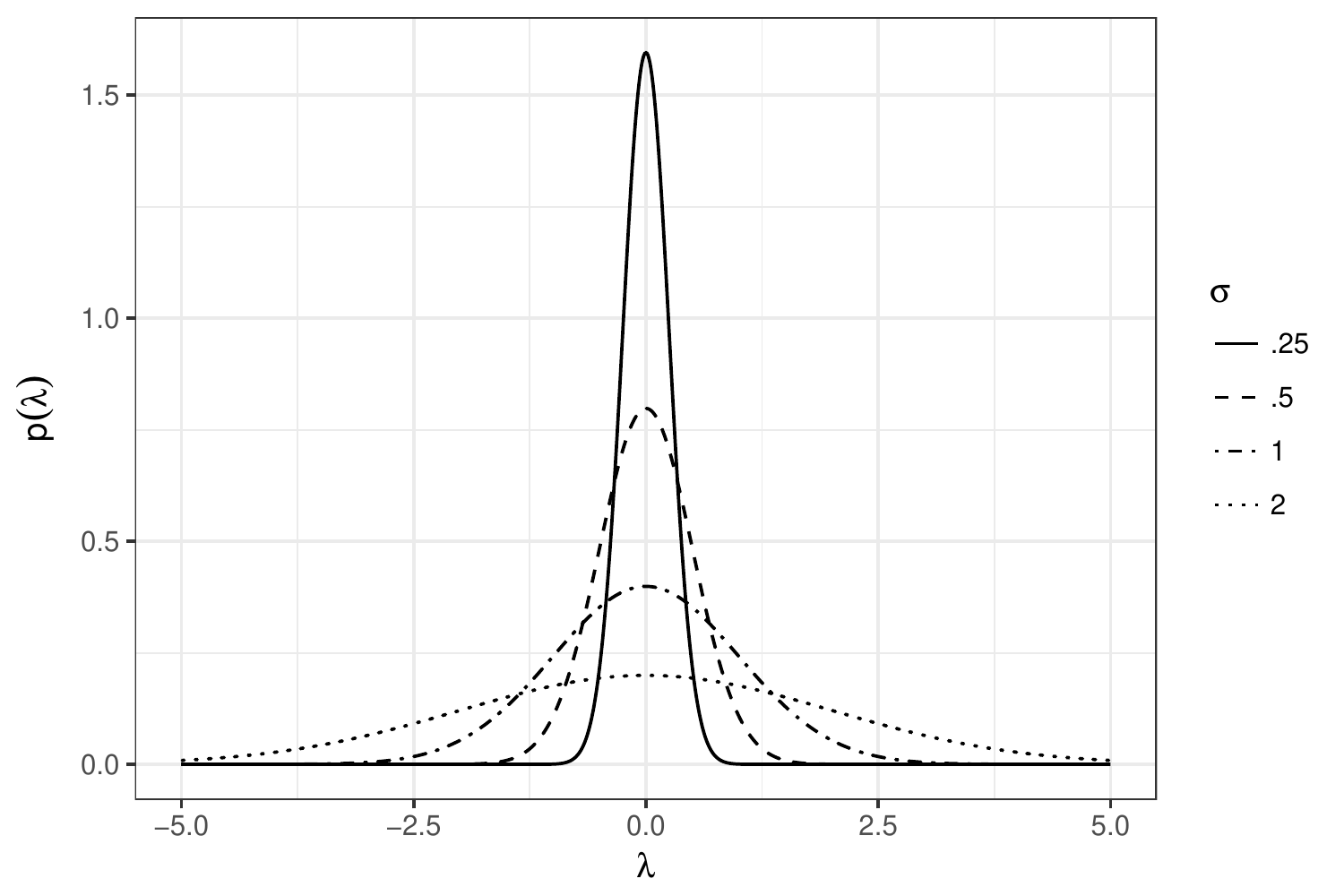}
        \caption{$p(\lambda_i \given I_{\mathcal{N}})$ for different values of $\sigma$ (we use $\sigma$ due to R's parameterization of the Gaussian distribution).}
        \label{fig:gp}
    \end{subfigure}
    ~           \begin{subfigure}[b]{0.5\textwidth}
        \includegraphics[width=\textwidth]{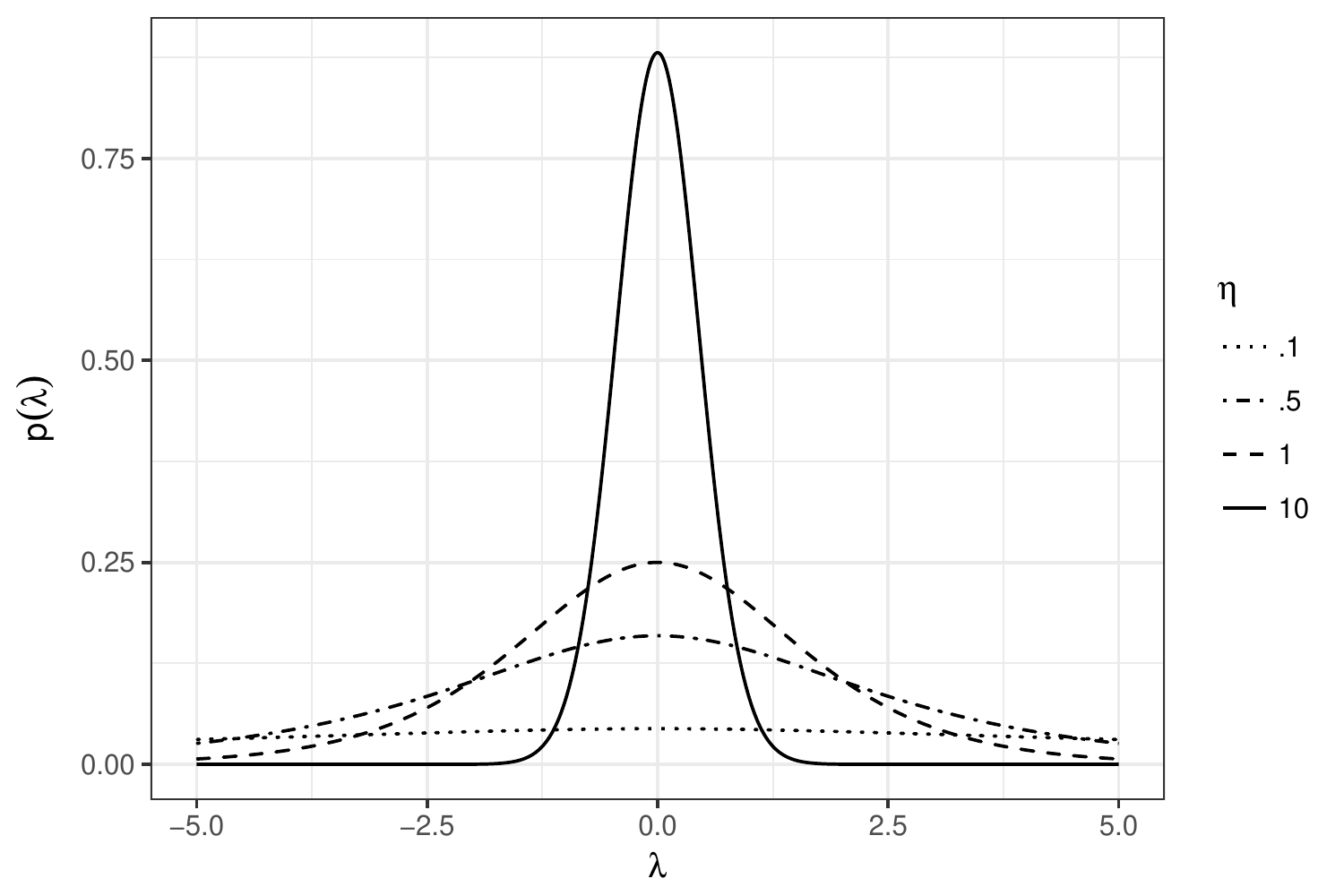}
        \caption{$p(\lambda_i \given I_{\beta})$ for different values of $\eta$. This family is known as the type III generalized logistic distribution.}
        \label{fig:tbp}
    \end{subfigure}
    \caption{Prior distributions in $\lambda_i$ that satisfy the desiderata described in \ref{desi}}
    \label{fig:priors}
\end{figure}

\section{The Model}
\label{motiv}
\subsection{Motivation} 
Recent advances in Bayesian computation mean that for models of a reasonable size,
we no longer have to restrict ourselves to using conjugate priors, Laplace approximations, or hand-tuned hyperparameters. Modern
MCMC methods such as Hamiltonian Monte Carlo (HMC) allow for rich hierarchical models that can be implemented easily in probabilistic programming languages such as Stan \cite{Stan}. These efficient
MCMC algorithms mean we can integrate over our uncertainty in the form of expectations, which, from a Bayesian perspective,
are preferable to optimization-based procedures \cite{Betancourt2}. Thus, we take the
stance that hierarchical models are the ideal way to approach Bayesian inference, especially now that the computational tools to exploit such models exist. The hallmark of hierarchical
modeling is to place additional priors over hyperparameters of interest. In our case, this will mean deriving a suitable prior distribution for $\eta$ or $\sigma$. \\ \\
The model is informed by our applications, which are discussed in the final section. We wish to exploit the advantages of Bayesian inference, while remaining objective
in our treatment of individual teams. This will lead us to a weakly-informative hierarchy that encompasses both
the objective and subjective approaches to the Bayesian paradigm. As seen in section \ref{hyper}, we incorporate prior information that pertains to the league as a whole at the uppermost layer
of the model, but insist that teams be evaluated only based on their performance against one another. Bradley-Terry provides an ideal framework for such a philosophy since it relies solely on
the $\binom{30}{2}$ head-to-head records. This affords us a rich framework for handling uncertainty, while keeping the model primarily data-driven.   

\subsection{Likelihood}
We restate the likelihood, this time in terms of $\boldsymbol{\lambda}$:
\begin{align}
  \begin{split}
  p(\mathbf{V} \given \boldsymbol{\lambda}) & = 
  \prod_{i=1}^{N-1} \prod_{j=i+1}^N \binom{n_{ij}}{V_{ij}} \left ( \frac{e^{\lambda_i}}{e^{\lambda_i} + e^{\lambda_j}} \right )^{V_{ij}} \left ( \frac{e^{\lambda_j}}{e^{\lambda_i}+e^{\lambda_j}} \right )^{V_{ji}} \\
  & \propto \prod_{i=1}^{N} \prod_{j=1}^N \left ( \frac{e^{\lambda_i}}{e^{\lambda_i} + e^{\lambda_j}} \right )^{V_{ij}}.
  \end{split}
\end{align}

\subsection{Choosing between $I_{\beta}$ and $I_{\mathcal{N}}$}
One initial complication is that there is no ``obvious'' way to choose between $I_{\mathcal{N}}$ and $I_{\beta}$. Figure \ref{fig:priors} illustrates the similarity between the densities of the
two families, a fact that is discussed in \cite{NE}. $I_{\mathcal{N}}$ is attractive on the grounds that Gaussians are often easy to work with, but we can also motivate its usage by appealing to the
principle of maximum entropy. Note that the variance of the $\lambda_i$ should by prior-independent; that is $\mathbb{V}[\lambda_i \given I_{\mathcal{N}}] = \mathbb{V}[\lambda_i \given I_{\beta}]$. 
Thus the first two moments are
fixed at $0$ and $\sigma^2$ respectively. It is well known that under these circumstances the differential entropy
\begin{equation}
  \mathbb{H}(\lambda_i) = - \int_{\mathbb{R}}  p(\lambda_i) \log p(\lambda_i)\mathrm{d}\lambda_i
\end{equation}
is maximized for $\lambda_i \sim \mathcal{N}(0, \sigma^2)$. Equivalently, this maximizes the relative entropy\footnote{The relative entropy is defined to be
  $\mathbb{H}_{R}(x) = - \int_{\Omega}  p(x) \log \left ( \frac{p(x)}{m(x)} \right ) \mathrm{d}x$, where $\Omega$ is the support space and $m$ is an invariant measure, meaning it transforms like
  $p$ under a change of variables.} under a uniform measure. So, in this sense $I_{\mathcal{N}}$ carries less information about $\lambda_i$, and we select it for this reason.

\label{hyper}

\subsection{Hyperparameters and Hyperpriors}
In hierarchical modeling, one foregoes the  hand-tuning of hyperparameters and instead builds another layer of prior distributions into the model, called hyperpriors. This creates the added difficulty of
determining good hyperpriors to use. Unfortunately, there is rarely a principled approach to determining this final layer of the model (we again reject the use improper priors,
to ensure the stability of MCMC methods \cite{Betancourt}). Often, the construction is made through some use of maximum likelihood estimation.
\\ \\ 
We will take the approach of using prior seasons' data to produce a
hyperprior for the hyperparameter $\sigma$.  In principle, one could
carry out MCMC on a previous season with a weakly informative
hyperprior, and produce a marginalized posterior for $\sigma$ which
could serve as a hyperprior for a subsequent season.  We opt for a
computationally-simpler approach based on an approximate maximum a
posteriori expansion.  The result will be a point estimate
$\widehat{\sigma}$ with an associated variance $\widehat{\varsigma}$.
Rather than a Gaussian approximation for the posterior on $\sigma$
(which would extend to negative values of $\sigma$), we instead choose a Gamma
distribution (using the shape and rate parameterization) with the same mean $\widehat{\sigma}$ and variance
$\widehat{\varsigma}$, i.e., $\sigma \sim \Gamma \left (
  \frac{\widehat{\sigma}^2}{\widehat{\varsigma}},
  \frac{{\widehat{\sigma}}}{\widehat{\varsigma}} \right)$.
  \\ \\
Formally assuming a uniform prior on $\sigma$, the log-posterior can be written
\begin{equation}
  \ell = \log p( \boldsymbol{\lambda},\sigma \given \mathbf{V} )
  = \sum_{i=1}^N
  \left\{
    \sum_{j=1}^N V_{ij}
    \left[
      \lambda_i - \log (e^{\lambda_i}+e^{\lambda_j})
    \right]
    - \frac{\lambda_i^2}{2\sigma^2}
  \right\}
  - N\log\sigma
  + \text{const.}
\end{equation}
and the MAP point can be found by taking the partial derivatives
\begin{equation}
  \frac{\partial\ell}{\partial\sigma}
  = \sigma^{-3} \sum_{i=1}^N \lambda_i^2
  - N\sigma^{-1}
  \hspace{8mm} \text{and} \hspace{8mm}
  \frac{\partial\ell}{\partial\lambda_i}
  = \sum_j \left(
    V_{ij}
    + n_{ij} \frac{e^{\lambda_i}}{e^{\lambda_i}+e^{\lambda_j}}
  \right)
  - \sigma^{-2} \lambda_i.
\end{equation}
Setting these to zero and rearranging produces the coupled MAP equations
\begin{equation}
  \widehat{\sigma} = \sqrt{\frac{1}{N}\sum_{i=1}^N \widehat{\lambda}_i^2}
  \hspace{8mm} \text{and} \hspace{8mm}
  \widehat{\lambda}_i = \log
  \left \{
    \frac{
      V_i-\widehat{\lambda}_i/\widehat{\sigma}^2
    }{
      \sum_{j=1}^N \left (
        n_{ij} \bigg /
        \left [
          e^{\widehat{\lambda}_i}
          + e^{ \widehat{\lambda}_j}
        \right ]
      \right )}
  \right \}
\end{equation}
where $V_i=\sum_{j=1}^N V_{ij}$ is the total number of games won by
team $i$.  These MAP equations could be solved iteratively by a method
analogous to that of Ford \cite{Ford}, but we make the assumption
that, with each team playing 162 games in a full season, $V_i$ is
large compared to $\widehat{\lambda}_i/\widehat{\sigma}^2$ and we can
use the maximum likelihood estimates
$\left \{\widehat{\lambda}_i^{\text{MLE}}\right \}$, determined by iteratively
solving
\begin{equation}
  \label{e:MLE}
  \widehat{\lambda}_i^{\text{MLE}} = \log \left \{ \frac{V_i}{ \sum_{j=1}^N \left ( n_{ij} \bigg / \left [ e^{\widehat{\lambda}_i^{\text{MLE}}} + e^{ \widehat{\lambda}_j^{\text{MLE}}} \right ]
      \right )}
  \right \}
  \ 
\end{equation}
in place of the $\left \{\widehat{\lambda}_i \right \}$, and writing
\begin{equation}
  \label{simple}
  \widehat{\sigma}
  \approx
  \sqrt{\frac{1}{N}\sum_{i=1}^N \left(\widehat{\lambda}_i^{\text{MLE}}\right)^2}.
\end{equation}
Note that the maximum likelihood equations \eqref{e:MLE} only
determine the $\left \{\widehat{\lambda}_i^{\text{MLE}} \right \}$
up to an overall additive constant, which we set by requiring
$\sum_{i=1}^N\widehat{\lambda}_i^{\text{MLE}}=0$. The variance $\widehat{\varsigma}$ can be estimated by considering the
matrix of second derivatives
\begin{equation}
  \mathbf{H}
  =
  \begin{pmatrix}
    -\frac{\partial^2\ell}{\partial^2\sigma}
    & \left \{-\frac{\partial^2\ell}{\partial\sigma\,\partial\lambda_j} \right \} \\
    \left \{-\frac{\partial^2\ell}{\partial\lambda_i\,\partial\sigma}\right \}
    &\left \{-\frac{\partial^2\ell}{\partial\lambda_i\,\partial\lambda_j}\right \}
  \end{pmatrix}_{\sigma=\widehat{\sigma};
    \,\boldsymbol{\lambda}=\widehat{\boldsymbol{\lambda}}}
\end{equation}
and defining $\widehat{\varsigma}=\left[\mathbf{H}^{-1}\right]_{\sigma\sigma}$.
The second derivatives are
\begin{subequations}
  \begin{gather}
    H_{\sigma\sigma}
    = 3\widehat{\sigma}^{-4} \sum_{i=1}^N\widehat{\lambda}_i^2
    - N\widehat{\sigma}^{-2} = 2N\widehat{\sigma}^{-2}
    \\
    H_{\sigma\lambda_i}
    = 2\widehat{\lambda}_i\widehat{\sigma}^{-3}
    \\
    H_{\lambda_i\lambda_j}
    = \delta_{ij} \left(
      \widehat{\sigma}^{-2}
      + \sum_{k=1}^N n_{ik}
      \frac{e^{\widehat{\lambda}_i+\widehat{\lambda}_k}}
      { e^{\widehat{\lambda}_i} + e^{ \widehat{\lambda}_k} }
    \right)
    - n_{ij}\frac{e^{\widehat{\lambda}_i+\widehat{\lambda}_j}}
    { e^{\widehat{\lambda}_i} + e^{ \widehat{\lambda}_j} }.
  \end{gather}
\end{subequations}
In practice, given a whole season's worth of data, we don't need to
invert the full matrix; the terms involving $\{n_{ij}\}$ will dominate
to leading order, and we can approximate the matrix as
block-diagonal\footnote{Note that the $\widehat{\sigma}^{-2}$ is
  important for inversion of the other block, which is otherwise
  degenerate since $\sum_{j=1}^N
  H_{\lambda_i\lambda_j}=\widehat{\sigma}^{-2}$.} to write
\begin{equation}
  \widehat{\varsigma} \approx \frac{1}{H_{\sigma\sigma}}
  = \frac{\widehat{\sigma}^2}{2N}.
\end{equation}
\\ \\
Following this procedure, we arrive at a $\Gamma \left( 2N,(2N) \big / {\widehat{\sigma}^2} \right)$ hyperprior. In keeping with the Bayesian philosophy, we avoid using the data from the season to be modeled
in setting that season's hyperprior. Instead, we account for any trends in league parity by constructing the hyperprior using the previous season's data. Since Major League Baseball has consisted of $30$
teams throughout the seasons we model, the hyperprior is $\Gamma \left( 60,60 \big / {\widehat{\sigma}^2} \right)$ where $\widehat{\sigma}^2$ is the estimated variance of the $\{\lambda_i\}$ during
the previous season.
  \begin{table}[ht]
    \label{sigmatable}
  \centering
  \caption{The estimates $\widehat{\sigma}$ and $\sqrt{\hat{\varsigma}}$ for the 2010 - 2016 seasons, computed according to the above prescription. We construct the hyperprior for a
  a given season using the estimates from the previous season. $\sqrt{\hat{\varsigma}}$
    can be interpreted as one standard deviation of uncertainty in $\widehat{\sigma}$.} 
\begin{tabular}{rcc}
  \hline
 year & $\widehat{\sigma}$ & $\sqrt{\hat{\varsigma}}$ \\ 
  \hline
  2010 & 0.264 & 0.034 \\
  2011 & 0.267 & 0.034 \\ 
  2012 & 0.316 & 0.041 \\ 
  2013 & 0.289 & 0.037 \\ 
  2014 & 0.235 & 0.030 \\ 
  2015 & 0.274 & 0.035 \\ 
  2016 & 0.262 & 0.034 \\ 
   \hline
\end{tabular}
\end{table}

  \begin{figure}
    \centering
        \includegraphics[width=.6\textwidth]{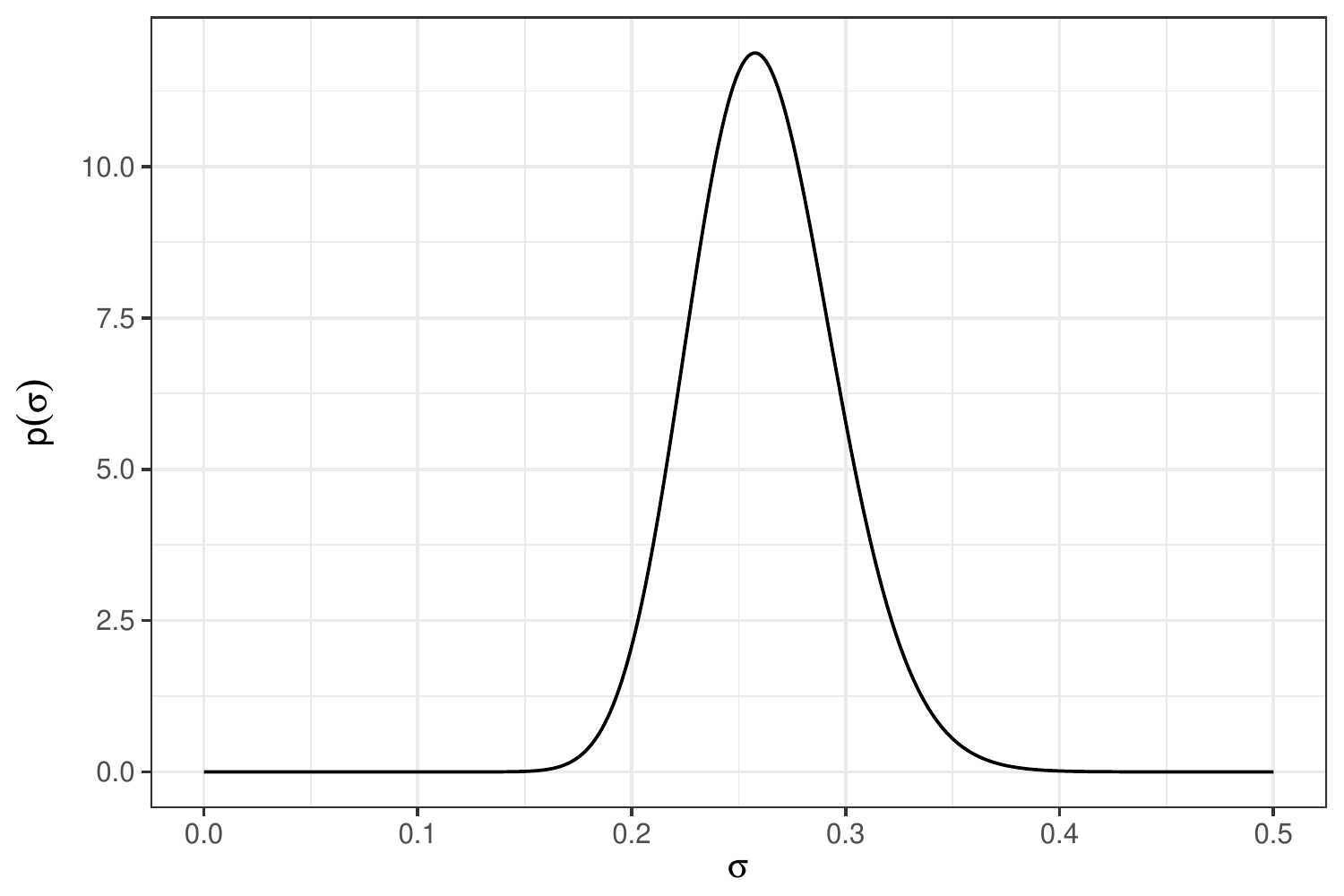}
        \caption{The hyperprior $p(\sigma)$ used to model the 2017 season, generated according to estimates from the 2016 season.}
    \label{fig:hyperpriors}
\end{figure}

\begin{table}[ht]
  \centering \caption{The estimates for $\left \{ \widehat{\lambda}_i^{\text{MLE}}\right \}$ and $\left \{ \widehat{\lambda}_i^{\text{MAP}}\right \}$ during the 2017 season,
    computed according to Ford's iterative algorithm. Note that
    the values are indeed close, justifying the simplification made in equation \ref{simple}. Teams in \textbf{boldface} made the postseason.}
\begin{tabular}[t]{rlrr}
  \hline
 & team & ${\widehat{\lambda}_i}^{\text{MLE}}$ & ${\widehat{\lambda}_i}^{\text{MAP}}$ \\ 
  \hline
 1 & \textbf{CLE} & 0.52 & 0.47 \\ 
 2 &  \textbf{HOU} & 0.51 & 0.46 \\ 
 3 &  \textbf{LAN} & 0.50 & 0.45 \\ 
 4 &  \textbf{BOS} & 0.33 & 0.29 \\ 
 5 &  \textbf{NYA} & 0.29 & 0.26 \\ 
 6 &  \textbf{WAS} & 0.26 & 0.25 \\ 
 7 & \textbf{ARI} & 0.26 & 0.25 \\ 
 8 &  \textbf{CHN} & 0.19 & 0.19 \\ 
 9 &  \textbf{MIN} & 0.13 & 0.12 \\ 
 10 & \textbf{COL} & 0.12 & 0.11 \\ 
 11 &  MIL & 0.07 & 0.08 \\ 
 12 & TBA & 0.05 & 0.03 \\ 
 13 & LAA & 0.03 & 0.02 \\ 
 14 & KCA & 0.02 & 0.02 \\ 
 15 & SLN & 0.00 & 0.01 \\
 \hline
\end{tabular}
\begin{tabular}[t]{rlrr}
  \hline
 & team & ${\widehat{\lambda}_i}^{\text{MLE}}$ & ${\widehat{\lambda}_i}^{\text{MAP}}$ \\ 
  \hline
 16 & SEA & -0.02 & -0.03 \\ 
 17 & TEX & -0.03 & -0.04 \\ 
 18 & TOR & -0.04 & -0.06 \\ 
 19 &  BAL & -0.06 & -0.07 \\ 
 20 & OAK & -0.09 & -0.10 \\ 
 21 & PIT & -0.18 & -0.15 \\ 
 22 & MIA & -0.19 & -0.15 \\ 
 23 &  SDN & -0.25 & -0.22 \\ 
 24 &  CHA & -0.27 & -0.26 \\ 
 25 &  ATL & -0.30 & -0.26 \\ 
 26 &  CIN & -0.33 & -0.29 \\ 
 27 &  NYN & -0.34 & -0.29 \\ 
 28 &  DET & -0.34 & -0.33 \\ 
  29 & SFN & -0.41 & -0.36 \\ 
  30 & PHI & -0.43 & -0.37 \\ 
   \hline
\end{tabular}
\end{table}

\subsection{Full Model}
The following describes the full Bayesian hierarchical model in matrix notation:
\begin{equation}
  \begin{split}
    \sigma & \sim \Gamma \left( 2N, \frac{2N}{\widehat{\sigma}^2} \right) \\
    \boldsymbol{\lambda} \given \sigma & \sim \mathcal{N}(\mathbf{0}, \sigma^2 \mathbf{I}) \\
    \mathbf{V} \given \boldsymbol{\lambda} & \sim \text{Bradley-Terry}(\exp\{\boldsymbol{\lambda}\}).
    \end{split}
\end{equation}
\section{Applications to Major League Baseball}
\label{apps}

\subsection{A Word on Data Acquisition and Computation}
The present authors have obtained all data from \verb=baseballreference.com= \cite{BR} and \verb=retrosheet.org= \cite{Retrosheet}.
Modifications of data and numerical computations were performed in the R and Stan programming languages \cite{R, Stan}. Stan is a probabilistic programming language for performing HMC.
HMC allows for efficient MCMC, capable of computing marginal posterior distributions for complex Bayesian models. Stan also permits the drawing of samples from the posterior predictive distribution.
For more information about Stan, see \cite{Stan}; for more information about HMC, see \cite{Betancourt}.

\subsection{Application I: Ranking Systems}
The first application we present is that of a ranking system based on the log-strengths. Such a system could generate weekly or monthly rankings more nuanced
than that provided by simple win-loss comparisons. Much of our motivation for treating the teams objectively is related to this application; a reliable ranking system should be based on team
performance alone. The Bayesian approach permits us to assign ranks based on $\mathbb{E}[\lambda_i \given \mathbf{V}]$, which integrates over possible outcomes rather than finding an optimum based on the data.
Table \ref{table:ranks} shows the final rankings from the 2017 season, with teams in \textbf{boldface} having made the postseason. Note that teams may be compared via the distance between their
respective log-strengths (which correspond to ratios in $\boldsymbol{\pi}$-space). In table \ref{table:mleranks}, we compare these results to those found by maximum likelihood estimation. This aptly
illustrates the effect of the Bayesian approach. The prior serves as a regularizer and promotes shrinkage, protecting against over-fitting. Unsurprisingly, $\mathbb{E}[\lambda_i \given \mathbf{V}]$ is more correlated with a team's true record than is $\widehat{\lambda}_i^\text{MLE}$.

\begin{figure}
    \centering
    \begin{subfigure}[b]{0.49\textwidth}
        \includegraphics[width=\textwidth]{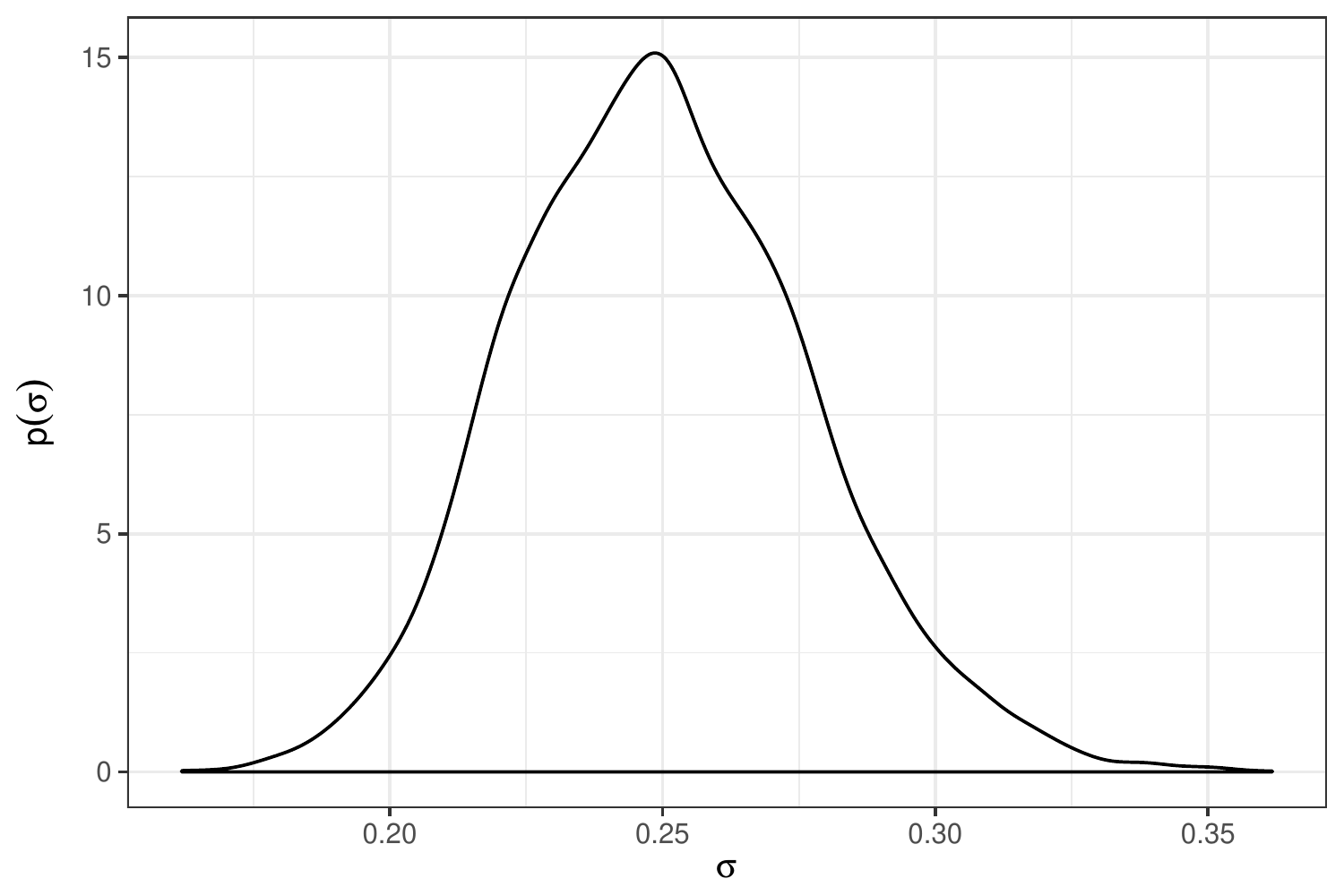}
        \caption{$p(\sigma \given \mathbf{V})$}
        \label{fig:sigmapost}
    \end{subfigure}
    \begin{subfigure}[b]{0.49\textwidth}
        \includegraphics[width=\textwidth]{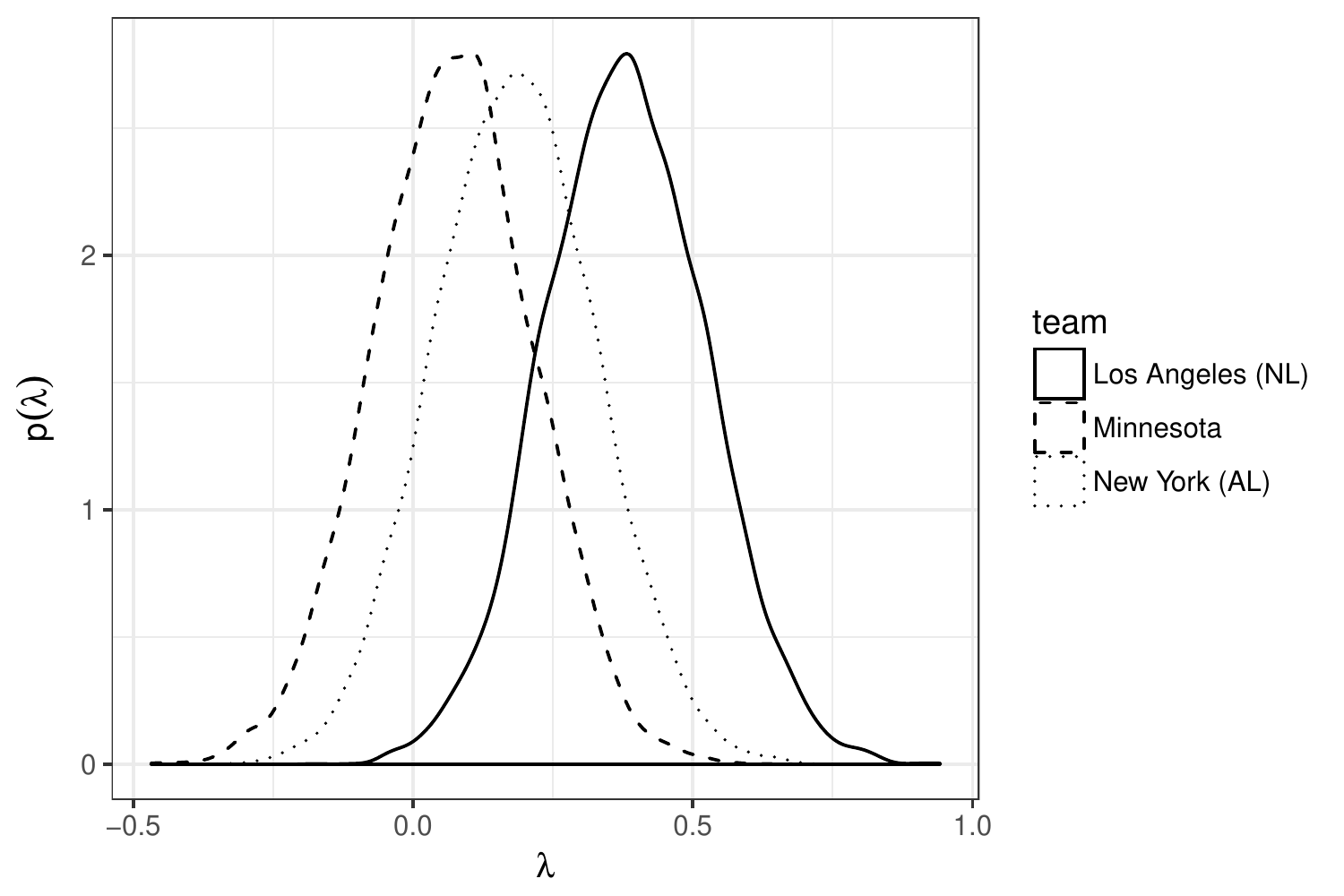}
        \caption{$p(\lambda_{\text{LAN}} \given \mathbf{V})$ vs. $p(\lambda_{\text{MIN}} \given \mathbf{V})$ vs. $p(\lambda_{\text{NYA}} \given \mathbf{V})$}
        \label{fig:comparepost}
    \end{subfigure}
    \caption{The left image shows the marginal posterior density for $\sigma$ during the 2017 season. As expected, it looks much like our informative hyperprior.
      The right image compares the marginal posterior densities of three postseason teams from 2017. The superiority of the NL-champion Dodgers is clear.}
    \label{fig:2016}
\end{figure}

\begin{table}[ht] \centering
    \caption{Final 2017 Major League Baseball rankings based on hierarchical Bayesian Bradley-Terry. Whereas $\widehat{\lambda}_i^{\text{MLE}}$ is found by optimizing the likelihood, 
      $\mathbb{E}[\lambda_i \given \mathbf{V}]$ is found by integrating over a posterior probability density.}
\begin{tabular}[t]{rlrr} 
  \hline
 & team & $\mathbb{E}[\lambda_i \given \mathbf{V}]$ & wins \\ 
  \hline
1 & \textbf{LAN} & 0.38 & 104 \\ 
  2 & \textbf{CLE} & 0.35 & 102 \\ 
  3 & \textbf{HOU} & 0.35 & 101 \\ 
  4 & \textbf{WAS} & 0.22 & 97 \\ 
  5 & \textbf{BOS} & 0.21 & 93 \\ 
  6 & \textbf{ARI} & 0.20 & 93 \\ 
  7 & \textbf{NYA} & 0.18 & 91 \\ 
  8 & \textbf{CHN} & 0.16 & 92 \\ 
  9 & \textbf{COL} & 0.10 & 87 \\ 
  10 & \textbf{MIN} & 0.07 & 85 \\ 
  11 & MIL & 0.07 & 86 \\ 
  12 & SLN & 0.02 & 83 \\ 
  13 & TBA & 0.00 & 80 \\ 
  14 & LAA & 0.00 & 80 \\ 
  15 & KCA & -0.01 & 80 \\ 
   \hline
\end{tabular}
\begin{tabular}[t]{rlrr}
  \hline
 & team &  $\mathbb{E}[\lambda_i \given \mathbf{V}]$ & wins \\ 
  \hline
  16 & SEA & -0.04 & 78 \\ 
  17 & TEX & -0.04 & 78 \\ 
  18 & TOR & -0.06 & 76 \\ 
  19 & BAL & -0.08 & 75 \\ 
  20 & OAK & -0.09 & 75 \\ 
  21 & MIA & -0.10 & 77 \\ 
  22 & PIT & -0.12 & 75 \\ 
  23 & SDN & -0.17 & 71 \\ 
  24 & ATL & -0.19 & 72 \\ 
  25 & NYN & -0.21 & 70 \\ 
  26 & CHA & -0.22 & 67 \\ 
  27 & CIN & -0.22 & 68 \\ 
  28 & DET & -0.27 & 64 \\ 
  29 & PHI & -0.28 & 66 \\ 
  30 & SFN & -0.28 & 64 \\ 

    \hline
\end{tabular}
\label{table:ranks}
\end{table}

  \begin{table}[ht] \centering
    \caption{Rankings from the 2017 season using maximum likelihood estimates. Maximum likelihood tends to produce results that diverge more
    from a team's actual record. }
\begin{tabular}[t]{rlrr}
  \hline
 & team & ${\widehat{\lambda}_i}^{\text{MLE}}$ & wins \\ 
  \hline
1 & \textbf{CLE} & 0.52 & 102 \\ 
  2 & \textbf{HOU} & 0.51 & 101 \\ 
  3 & \textbf{LAN} & 0.50 & 104 \\ 
  4 & \textbf{BOS} & 0.33 & 93 \\ 
  5 & \textbf{NYA} & 0.29 & 91 \\ 
  6 & \textbf{WAS} & 0.26 & 97 \\ 
  7 & \textbf{ARI} & 0.26 & 93 \\ 
  8 & \textbf{CHN} & 0.19 & 92 \\ 
  9 & \textbf{MIN} & 0.13 & 85 \\ 
  10 & \textbf{COL} & 0.12 & 87 \\ 
  11 & MIL & 0.07 & 86 \\ 
  12 & TBA & 0.05 & 80 \\ 
  13 & LAA & 0.03 & 80 \\ 
  14 & KCA & 0.02 & 80 \\ 
  15 & SLN & 0.00 & 83 \\ 

   \hline
\end{tabular}
\begin{tabular}[t]{rlrr}
  \hline
 & team &  ${\widehat{\lambda}_i}^{\text{MLE}}$ & wins \\ 
  \hline
  16 & SEA & -0.02 & 78 \\ 
  17 & TEX & -0.03 & 78 \\ 
  18 & TOR & -0.04 & 76 \\ 
  19 & BAL & -0.06 & 75 \\ 
  20 & OAK & -0.09 & 75 \\ 
  21 & PIT & -0.18 & 75 \\ 
  22 & MIA & -0.19 & 77 \\ 
  23 & SDN & -0.25 & 71 \\ 
  24 & CHA & -0.27 & 67 \\ 
  25 & ATL & -0.30 & 72 \\ 
  26 & CIN & -0.33 & 68 \\ 
  27 & NYN & -0.34 & 70 \\ 
  28 & DET & -0.34 & 64 \\ 
  29 & SFN & -0.41 & 64 \\ 
  30 & PHI & -0.43 & 66 \\ 

    \hline
\end{tabular}
\label{table:mleranks}
\end{table}
    
\subsection{Application II: Predictive Modeling}
Bayesian probability offers a particularly elegant way of handling prediction. For our model, the posterior predictive distribution is given by
\begin{align}
  p(\tilde{\mathbf{V}} \given \mathbf{V}) = \int_{\mathbb{R}^N}  p(\tilde{\mathbf{V}} \given \boldsymbol{\lambda})p(\boldsymbol{\lambda} \given \mathbf{V}) \diff \boldsymbol{\lambda},
\end{align}
which integrates over all uncertainty in the model and gives a distribution over unobserved data $\tilde{\mathbf{V}}$ conditional on the observed data
$\mathbf{V}$. The point estimate used for prediction is $\mathbb{E}[\tilde{\mathbf{V}} \given \mathbf{V}]$.
The accuracy of the predictive distribution can be readily measured by splitting the sample. We fit the data in a given season up to a certain date, and predict team records for the remainder of the
season. This can of course be validated against the known outcomes. In
machine learning terminology, the date at which we partition the data represents the separation between the training and the test sets. We can define a loss function, or error metric, to
evaluate the overall validity of this approach, and compare it to predictions based on generating samples from maximum likelihood estimates alone. The respective error metrics for a given
team are
\begin{equation}
  \begin{split}
    \mathrm{error}_i^{\text{Bayes}} & = \left | \mathbb{E} \left [ \tilde{V}_i^{\text{test}} \given V_i^{\text{train}} \right ] - V_i^{\text{test}} \right | \\
    \mathrm{error}_i^{\text{MLE}} & = \left | \mathbb{E} \left [ \tilde{V}_i^{\text{test}} \hspace{1mm};\hspace{1mm} \widehat{\lambda}_i^{\text{train}} \right ] - V_i^{\text{test}} \right |.
  \end{split}
\end{equation}
In words, these are the absolute distances between the predicted wins and actual wins in the test set. An overall error metric can be given by
\begin{equation}
  \begin{split}
    \mathrm{error}^{\text{Bayes}} & = \frac{1}{N} \sum_{i=1}^{N} \left ( \mathrm{error}_i^{\text{Bayes}} \right ) \\
    \mathrm{error}^{\text{MLE}} & = \frac{1}{N} \sum_{i=1}^{N} \left ( \mathrm{error}_i^{\text{MLE}} \right ), \\
  \end{split}
\end{equation}
the means of the individual error metrics. Similarly,
\begin{equation}
  \begin{split}
    \mathrm{sd}^{\text{Bayes}} & = \sqrt{ \frac{1}{N} \sum_{i=1}^{N} \left ( \mathrm{error}_i^{\text{Bayes}} - \mathrm{error}^{\text{Bayes}}\right )^2 } \\
    \mathrm{sd}^{\text{MLE}} & = \sqrt{ \frac{1}{N} \sum_{i=1}^{N} \left ( \mathrm{error}_i^{\text{Bayes}} - \mathrm{error}^{\text{MLE}}\right )^2 } \\
  \end{split}
\end{equation}
measures how variable these estimates are. The results for the 2017 season are shown in table \ref{table:errors}. In figure \ref{fig:predictions}, we
plot at each partition date the average overall predictive metrics during the seasons 2011 - 2017. With sufficient data, the two methods perform similarly.
However, the Bayesian approach offers much-improved performance when data is scarce. 
This is true for both error and variability. In this sense, our model is preferable early in the season,
and continues to outperform MLE-based predictions into July, after which the two methods begin to converge in
accuracy. Without doubt, higher accuracy could be achieved under different approaches if that were the sole goal; we aim to strike a balance between our two described applications.
\begin{table}
  \caption{Comparison of error rates from predictions based on hierarchical Bayesian Bradley-Terry and maximum likelihood for the 2017 season.
    ``Partition'' indicates where the data was split
    into a training and test. The Bayesian approach performs significantly better during the first half of the season.}
\centering
\begin{tabular}{crrrrr}
  \hline
 partition & $\mathrm{error}^{\text{Bayes}}$  & $\mathrm{error}^{\text{MLE}}$ & $\mathrm{sd}^{\text{Bayes}}$ & $\mathrm{sd}^{\text{MLE}}$ \\ 
  \hline
Apr15 & 8.82 & 24.65 & 6.58 & 17.34 \\ 
  May1 & 7.31 & 12.49 & 6.17 & 10.39 \\ 
  May15 & 6.20 & 9.84 & 5.68 & 5.84 \\ 
  Jun1 & 4.72 & 6.90 & 4.87 & 4.27 \\ 
  Jun15 & 4.32 & 4.81 & 4.65 & 4.79 \\ 
  Jul1 & 4.04 & 4.17 & 3.46 & 3.43 \\ 
  Jul15 & 3.90 & 4.01 & 3.72 & 3.96 \\ 
  Aug1 & 3.58 & 3.89 & 3.17 & 3.15 \\ 
  Aug15 & 3.32 & 3.26 & 2.74 & 2.91 \\ 
  Sep1 & 2.57 & 2.59 & 2.31 & 2.39 \\ 
  Sep15 & 1.75 & 1.83 & 1.14 & 1.11 \\ 

   \hline
\end{tabular}
\label{table:errors}
\end{table}
\newpage
  \begin{figure}
    \centering
    \begin{subfigure}{.8\textwidth}
        \includegraphics[width=\textwidth]{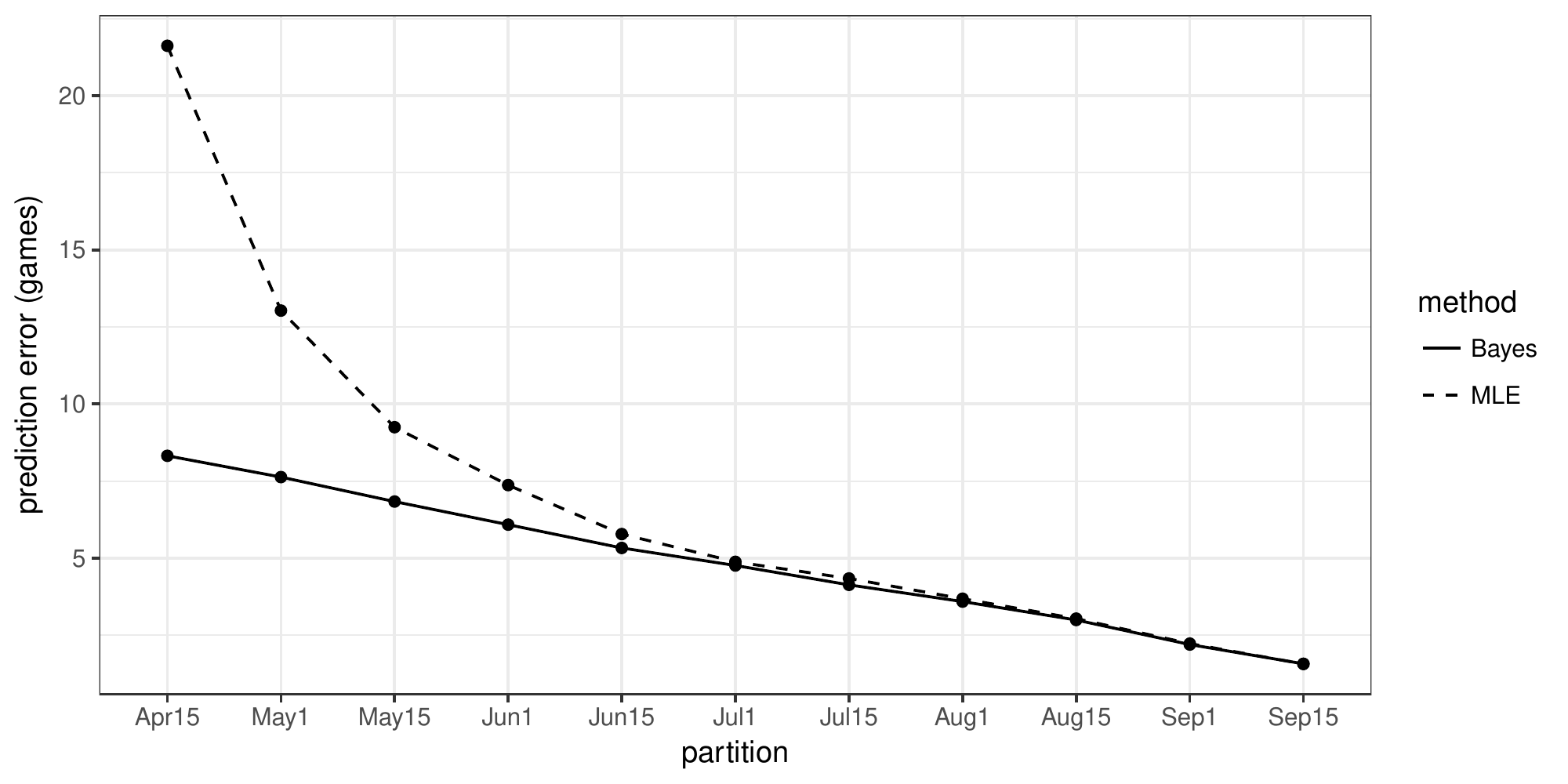}
                    \end{subfigure}
      \begin{subfigure}{.8\textwidth}
        \includegraphics[width=\textwidth]{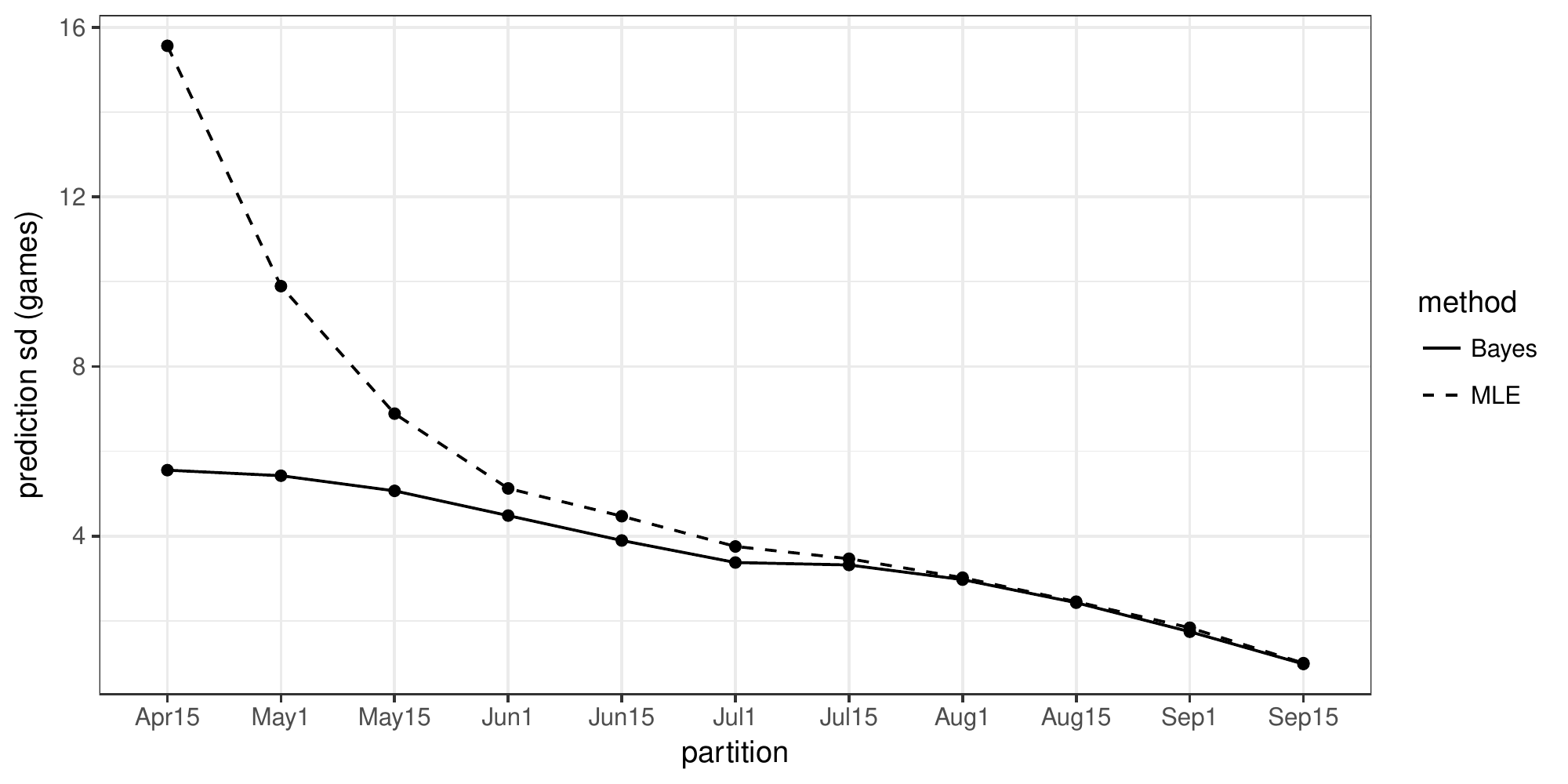}
                    \end{subfigure}  
      \caption{A comparison of predictions based on $\mathbb{E}[\tilde{\mathbf{V}}^{\text{test}} \given \mathbf{V}^{\text{train}}]$ and
        $\mathbb{E}[\tilde{\mathbf{V}}^{\text{test}} \hspace{.5mm};\hspace{.5mm}\widehat{\lambda}_i^{\text{train}}]$,
        averaged across the seasons 2011 - 2017.
        ``Partition'' indicates when the data was split into a training and test set. In general, Bayesian Bradley-Terry matches or beats the performance of MLE-based prediction for the entirety of a season,
      in terms of both error rate and error variability.}
    \label{fig:predictions}
  \end{figure}

\section{Conclusions}
Our proposed Bayesian Bradley-Terry model provides a useful and coherent framework for assessing and predicting the performance of Major League Baseball teams. By adhering to Whelan's desiderata \cite{Whelan},
we construct a hierarchical model that is weakly informative at the level of individual teams, but includes prior knowledge with respect to the entire league. This permits a model
that combines the subjective and objective approaches to Bayesian inference, capable of for use in a range of applications. Specifically, we demonstrate the merit of Bayesian Bradley-Terry
in applications to ranking and prediction, finding a balance between inferring latent structure and making respectable forecasts.
In both cases, our model outperforms maximum likelihood estimation by integrating over uncertainty, preventing over-fitting.
In summary, hierarchical Bayesian Bradley-Terry offers good performance in application, while being simple, interpretable, and compliant with the desirable properties of the desiderata.

\bibliographystyle{acm}
\end{document}